\documentclass[AMA,Times1COL]{WileyNJDv5} 

\usepackage{graphicx}
\articletype{Article Type}%

\received{Date Month Year}
\revised{Date Month Year}
\accepted{Date Month Year}
\journal{Journal}
\volume{00}
\copyyear{2023}
\startpage{1}

\begin{document}

\title{Outcome-Assisted Multiple Imputation of Missing Treatments}

\author[1]{Joseph Feldman}

\author[2]{Jerome P. Reiter}

\authormark{FELDMAN and REITER}
\titlemark{Outcome-Assisted Multiple Imputation of Missing Treatments}

\address[1]{\orgdiv{Department of Statistical Science}, \orgname{Duke University}, \orgaddress{\state{North Carolina}, \country{USA}}}

\corres{Joseph Feldman. \email{joseph.feldman@duke.edu}}


\abstract[Abstract]{We provide guidance on multiple imputation of missing at random treatments in observational studies. Specifically, analysts should account for both covariates and outcomes, i.e., not just use propensity scores, when imputing the missing treatments.  To do so, we develop outcome-assisted multiple imputation of missing treatments: the analyst fits a regression for the outcome on the treatment indicator and covariates, which is used to sharpen the predictive probabilities for missing treatments under an estimated propensity score model. 
We derive an expression for the bias of the inverse probability weighted estimator for the average treatment effect under multiple imputation of missing treatments, and we show theoretically that this bias can be made small by using outcome-assisted multiple imputation. 
Simulations demonstrate empirically that outcome-assisted multiple imputation can offer better inferential properties than using the treatment assignment model alone. 
We illustrate the procedure in an analysis of data from the National Longitudinal Survey of Youth.}



\maketitle

\renewcommand\thefootnote{}

\renewcommand\thefootnote{\fnsymbol{footnote}}
\setcounter{footnote}{1}
\section{Introduction}

In observational studies, the treatment status of some study subjects may be unknown.  For example, the treatment of interest may be an environmental exposure or health-related behavior that is measured only for a subset of the study subjects \citep{lee2021framework, mitra2023latent, feldman2024gaussian, chen2025nonparametric}. Or, when the data for the observational study come from a survey, some individuals may not respond to the question that defines the treatment status \citep{molinari2010missing}.
In such cases, analysts can improve the accuracy of causal inferences by accounting for (i.e., not disregarding) the missingness in the treatment status, especially when the treatments are not missing completely at random.



A flexible and convenient way to handle the missing treatments is to use multiple imputation \citep{mitra2023latent, shan:etal}.  The analyst estimates a model for treatment status using the available data. Using this model, the analyst generates multiple, plausible values of the missing treatments to create completed datasets. In each completed dataset, the analyst applies their causal inference procedure of choice, such as propensity score matching or inverse probability weighting, and combines the results via the multiple imputation combining rules \citep{rubin:1987}.

In practice, these models almost certainly  generate imputations that do not equal the corresponding true treatment statuses. 
While perfectly correct imputations may not be necessary for accurate inferences, 
poorly specified imputation models  can lead to 
inaccurate inferences \citep{mitra2023latent}.
As such, it is critical to use imputation models that capture essential features of the distribution of the missing treatments.  Imputation of missing treatments also forces the analyst to take a side on whether or not to use the outcome variables in the imputation model, with not using outcomes consistent with one interpretation of a ``design-first'' philosophy for causal inference 
\cite{rubin:design, imbens:rubin}.
Despite these complexities, as far as we are aware, there are few published works targeting multiple imputation for missing treatments.  

In this article, we present practical guidance for specifying multiple imputation models when treatments are missing at random. Using theoretical results, we argue that the imputation models for missing treatments generally should condition on both the covariates and outcomes to enable the most accurate inferences. Further, this theory suggests using a factorization
that employs a model for the treatment assignment  mechanism given covariates coupled with a model for the outcomes given  treatment and covariates. Importantly, the outcome model is used only to assist the imputation of missing treatments.  In each completed dataset, the analyst  estimates the treatment effect with their preferred causal estimator, which we do here using inverse probability weighting.
We refer to the entire procedure as Outcome-assisted Multiple Imputation of Treatments (OMIT).




The remainder of this article is organized as follows.  In Section \ref{prelim}, we review the causal inference setting that we work under.  As part of the review, we examine conditions under which complete-case analyses based on inverse probability weighting can provide valid estimates of treatment effects. We also discuss why multiple imputation approaches can be particularly appealing for causal inference in this setting as opposed to inverse probability weighting alone.  In Section \ref{sec:theory}, we derive a general expression for the bias of the multiple imputation point estimator using inverse probability weighting when treatments are missing.  We then explain why controlling for the outcome variable can make this bias go to zero. In Section \ref{sec:sim}, we present results of simulation studies that show the benefits of using OMIT.  The simulation studies suggest that OMIT can have a degree of robustness to outcome model misspecification. A key insight from the simulations, as well as the theoretical results, is that the outcome model can serve to sharpen the accuracy of the treatment imputations, even when it is (not grossly) misspecified. In Section \ref{sec:NLSY}, we apply OMIT to an illustrative application of inverse probability weighting using data from the National  Longitudinal Survey of Youth.  Finally, in Section \ref{sec:conclusion}, we summarize the key findings and identify several topics for future research.


\section{Preliminaries}\label{prelim}

\subsection{Causal Inference under the Potential Outcomes Framework}
We work in the potential outcomes framework to estimate causal effects \citep{rubin1974estimating}. For each unit $i$ in a study population with $n$ units, let $t_i = 1$ and $t_i=0$ indicate assignment to the treated and control conditions, respectively. For $i=1, \dots, n$, let $y_{i1}$ and $y_{i0}$  be the outcome under the treatment condition and the control condition, respectively.  We observe $y_i = t_iy_{i1} + (1-t_i)y_{i0}$, i.e., only one outcome is realized.  We adhere to the stable unit treatment value assumption (SUTVA) and strong ignorability \citep{rosenbaum1983central}; given a $d \times 1$ vector of covariates  $\boldsymbol{x}_i$, for all units $i$ we have $0 < p(t_i =1 \mid \boldsymbol x_i) < 1 $ and $(y_{i1},y_{i0}) \perp t_i \mid \boldsymbol x_i$. 

We focus on estimation of the average treatment effect (ATE) in the study population. 
For $i=1, \dots, n$, let $\tau_i = y_{i1} - y_{i0}$ be the treatment effect for unit $i$.
The ATE is given by 
\begin{equation}\label{ATE}
    \tau = n^{-1}\sum_{i=1}^{n} \tau_{i}.
\end{equation}
For $i=1, \dots, n$, let $e(\boldsymbol{x}_i) = p(t_i = 1 \mid \boldsymbol{x}_i)$ be the unit's propensity score, i.e., the probability of receiving treatment given the covariate profile $\boldsymbol x_i$\citep{rosenbaum1983central}.  
{The analyst estimates each $e(\boldsymbol{x}_i)$ typically using a binary regression of the treatments on the covariates $\boldsymbol x$.} For $i=1, \dots, n$, we abbreviate the estimated propensity score as $\hat{e}_i$.

To estimate $\tau$, we consider 
the inverse probability weighting  (IPW) estimator,
\begin{equation}\label{IPW}
    \hat{\tau} = n^{-1}\sum_{i=1}^{n} \frac{t_i y_i}{\hat{e}_i}  - \frac{(1-t_i)y_i}{(1-\hat{e}_i)}.
\end{equation}
To estimate \eqref{IPW}, analysts can use the \verb|survey| package\citep{lumley2020package} in \verb|R|, assuming an unequal probability sampling design.  This package also provides estimated standard errors, which we use in the simulations and NLSY analysis. 
Other approaches to estimating standard errors include sandwich estimators \citep{robins1994estimation}, asymptotic normal approximations \citep{lunceford2004stratification}, and bootstrapping approaches.   

We focus on the IPW estimator mainly because of its popularity but also in part because it facilitates the theoretical analyses in Section \ref{sec:theory}. 
The benefits of OMIT should extend to inferences constructed using other propensity-score based methods, including using matching \citep{rosenbaum1985constructing}, sub-classification \citep{rosenbaum1984reducing}, and balancing weights \citep{li2018balancing}.



\subsection{Effect of Missing Treatment Data on IPW Estimates}\label{sec:cc}

When treatments are missing, a naive application of the IPW estimator in \eqref{IPW} using only the complete cases (CC) can result in biased inferences for $\tau$.  We state this formally here as part of our review and illustrate it in the simulation studies.

For $i=1, \dots, n$, let $r_i = 1$ when $t_i$ is missing and $r_i = 0$ when $t_i$ is observed. 
The vectors of observed and missing treatment assignments are $\boldsymbol t_{obs} = \{t_{i}: r_i= 0\}$ and 
$\boldsymbol t_{mis} = \{t_i: r_i = 1\}$, respectively. For simplicity, here we presume that $\boldsymbol{x}_i$ and $y_i$ are completely observed for all units $i=1, \dots, n$. Defining $(\boldsymbol y, \boldsymbol r, \boldsymbol x) =\{y_i, r_i, \boldsymbol x_i\}_{i=1}^{n}$, the observed data are $(\boldsymbol y, \boldsymbol r, \boldsymbol x, \boldsymbol t_{obs})$.  We consider scenarios where the treatments are missing at random (MAR) \citep{rubin1976inference}, so that $p(r_{i} = 1 \mid \boldsymbol x_i, y_i, t_i) = p(r_{i} = 1 \mid \boldsymbol x_i, y_i)$, i.e., the probability of observing the treatment assignment depends on the  observed data.
Let $n_{obs} = \sum_{i=1}^n \boldsymbol{1}_{r_{i} = 0}$, where the indicator function $\boldsymbol{1}_a=1$ when the condition represented by $a$ is true and $\boldsymbol{1}_a=0$ otherwise. The CC version of \eqref{IPW} is given by
\begin{equation}\label{IPW-CC}
    \hat{\tau}_{CC} = n_{obs}^{-1}\sum_{i: r_{i} = 0} \frac{t_i y_i}{\hat{e}_i}  - \frac{(1-t_i)y_i}{(1-\hat{e}_i)}.
\end{equation}

In general, \eqref{IPW-CC} is a biased estimator of $\tau$, so that $\hat{\tau}_{CC}$ can be inaccurate even when $n_{obs}$ is large.  There are two general situations where this may not be the case. To see this, suppose for convenience that we have oracle, i.e., true, values of $e_{i}$ for all units in the complete cases and use them in place of $\hat{e}_i$ in \eqref{IPW-CC}. 
As stated in Proposition \ref{biased-CC}, which is proved in the supplementary material, in this setting \eqref{IPW-CC} offers unbiased estimates of $\tau$ in two specific scenarios.
\begin{proposition}\label{biased-CC} Assume strong ignorability and SUTVA, and that $p(\sum_{i=1}^{n} \boldsymbol{1}(r_i = 1) = n) = 0$ for any $n$. Assume we use the true propensity scores in $\hat{\tau}_{CC}$ as defined in \eqref{IPW-CC}. 
 If one of the following conditions holds:
\begin{itemize}
    \item there are homogeneous conditional treatment effects, i.e., $\tau_{i} = \tau_{j}, \forall i \neq j \in \{1, \dots, n\}$,
    \item treatments are missing completely at random (MCAR), i.e., 
    $p(r_{i} = 1 \mid \boldsymbol x_i) = p$ \ for \ $i=1, \dots, n$,
\end{itemize}
then $E(\hat{\tau}_{CC}) = \tau$, where the expectation is taken with respect to the joint distribution of the treatment assignment and missingness mechanism.
\end{proposition}

The assumption $p(\sum_{i=1}^{n} \boldsymbol{1}(r_i = 1) = n) = 0$ ensures that \eqref{IPW-CC} is always estimable.  Both the homogeneous treatment effects and MCAR assumptions 
are restrictive and may not be satisfied in practice. The first condition is violated when the covariates moderate individual treatment effects, which easily could be the case. 
The second condition
seems unlikely,
as missingness in the treatment variable may be explained by observed covariates. For example, and previewing the NLSY analysis, suppose 
individuals with certain demographic profiles 
are less likely to respond to the question that defines treatment or control status in the causal study.  
Then, the complete case data comprise a sub-population, so that CC estimates of $\tau$ may not generalize to the full study population.




\subsection{Multiple Imputation}\label{sec:MI}

In multiple imputation, the analyst estimates predictive models for the missing data given the observed data to create $M$ completed datasets, $D = (D^{(1)}, \dots, D^{(M)})$. In each $D^{(m)}$ where $m=1, \dots, M$, the analyst estimates the quantity of interest, in our case $\tau$, using some complete-data estimator, in our case $\hat{\tau}^{(m)}$.  
The analyst also estimates a complete-data variance $u^{(m)}$  corresponding to  $\hat{\tau}^{(m)}.$  For causal inference using the IPW estimator, each 
$\hat{\tau}^{(m)}$ is the value of  \eqref{IPW} computed with $D^{(m)}$, which we write explicitly  in \eqref{IPWmi} in Section \ref{sec:biasMI}, and $u_i^{(m)}$ is its accompanying 
estimated variance.  The analyst combines the results using the quantities,
\begin{eqnarray}
\bar{\tau} &=& M^{-1}\sum_{m=1}^M\hat{\tau}^{(m)}\label{eq:mipt}\\
\bar{u} &=& M^{-1}\sum_{m=1}^M u^{(m)} \label{eq:miubar}\\
b &=& (M-1)^{-1}\sum_{m=1}^M (\hat{\tau}^{(m)} - \bar{\tau})^{2}.\label{eq:mib}
\end{eqnarray}
The analyst makes inferences using $\bar{\tau} - \tau \sim t_{\nu}(0, T_M)$, where the variance $T_M = (1+1/M)b+\bar{u}$ and degrees of freedom $\nu = (M-1)(1+\bar{u}/((1+1/M)b))^2$. 
We note that the literature suggests multiple ways to compute propensity scores with multiple imputation  
\citep{mitra2016comparison}. 
We discuss our approach for OMIT in Section \ref{sec:theory} and Section \ref{sec:sim}.

As an alternative to multiple imputation with \eqref{IPWmi}, analysts can adjust the IPW estimator in \eqref{IPW-CC} to account for missingness in the treatments \citep{zhao2024covariate,zhang2016causal, williamson2012doubly}. The general idea is to 
multiply each estimated $\hat{e}_i$ by an estimate of $p(r_i=0 \mid \boldsymbol x_i, y_i)$. 
However, IPW with multiple imputation has some potential advantages over IPW with adjustments.  First, the adjustments to the IPW weights for missingness in $t$ can result in large adjusted weights, particularly when some of the probabilities of $r_i=0$ are near zero.  This in turn can lead to inflated variances; these need not arise when handling missing treatments using multiple imputation. For that matter, using the adjusted IPW requires estimating a model for $r_i$, whereas this may not be necessary when using IPW with multiple imputation. Second, multiple imputation offers flexibility in that, among other things, it can be applied readily with missing values also in covariates or outcomes.  These completed datasets 
can be used for multiple purposes, including analyses with different subsets of units or different causal estimands.  Third, it offers a prescriptive approach for obtaining uncertainty estimates that accounts for the missing data.  Fourth, it can be relatively robust to model misspecification. The imputation models only influence the cases with missing values, whereas adjusting the original IPW weights for a modeled missing data mechanism affects all observed records' contributions to the causal estimate.  

Another alternative is to bound values of the treatment effect estimate over all hypothetically possible values of the missing treatments \cite{molinari2010missing, mebane2013causal, kennedy2020efficient}. These non-parametric identification bounds can be wide in practice.


\section{Outcome-assisted Multiple Imputation of Treatments}\label{sec:theory}

Of course, to realize the benefits of multiple imputation, analysts require reasonable imputation models, to which we now turn.  We begin by deriving an expression for the bias of \eqref{eq:mipt} when each $\hat{\tau}^{(m)}$ is computed using \eqref{IPWmi}.

\subsection{Bias of Multiple Imputation Point Estimator} \label{sec:biasMI}


For $i=1, \dots, n$ in $D^{(m)}$, let $t^*_i$ be the imputed value of $t_i$ when $r_i=1$ and let $t_i^* = t_i$ when $r_i=0$. 
In each $D^{(m)}$ we compute    
\begin{eqnarray}
\hat{\tau}^{(m)} &=& n^{-1}\sum_{i=1}^n \frac{y_{i}t_i^*}{\hat{e}_{i}}-  \frac{y_{i}(1-t_i^*)}{(1-\hat{e}_{i})} \label{IPWmi}
\end{eqnarray}
 Because each $\hat{\tau}^{(m)}$ results from independent sampling from the imputation model, it suffices to examine the bias of  $\hat{\tau}^{(m)}$ for a single $m$. 
We can separate \eqref{IPWmi} into sums over the units where the imputed treatments match the true treatment status and where they do not. Recognizing that $y_i = y_{i1}$ when $(t_i^* = 0, t_i=1)$ and $y_i = y_{i0}$ when $(t_i^* = 1, t_i=0)$, we have 
\begin{eqnarray}
\hat{\tau}^{(m)} &=& (1/n)\left(\sum_{i: r_i=0} \frac{y_{i1}t_i}{\hat{e}_{i}} + \sum_{i: r_i=1, t_i=1, t_i^*=1} \frac{y_{i1}t_i}{\hat{e}_{i}} + \sum_{i: r_i=1, t_i=0, t_i^*=1} \frac{y_{i0}(1-t_i)}{\hat{e}_{i}}\right) \nonumber\\
&-& (1/n)\left(\sum_{i: r_i=0} \frac{y_{i0}(1-t_i)}{(1-\hat{e}_{i})} + \sum_{i: r_i=1, t_i=0, t_i^*=0} \frac{y_{i0}(1-t_i)}{(1-\hat{e}_{i})} + \sum_{i: r_i=1, t_i=1, t_i^*=0} \frac{y_{i1}t_i}{(1-\hat{e}_{i})}\right).\label{eq2}
\end{eqnarray}

To compute the bias,  we work with the finite population treatment effect in \eqref{ATE}.  The  potential outcomes and covariates are fixed, and the source of randomness is the process that assigns the true treatments, creates missingness in the treatments, and generates the imputations.
For convenience, to derive an expression for the bias 
we presume known and use the true values of 
$e_{i}$ for all $i=1, \dots, n$; thus, we replace $\hat{e}_i$ with $e_i$ in \eqref{eq2}.  In our simulations and data analyses, we estimate $\hat{e}_{i}$
using the observed data.  

For $i=1, \dots, n$, let $p_i = p(r_i =1 \mid \boldsymbol{x}_i, y_{i}, t_i)$
be the probability that the treatment assignment for unit $i$ is missing. Because we presume  MAR treatments, each $p_i= p(r_i =1 \mid \boldsymbol{x}_i, y_i)$. In addition, for $i=1, \dots, n$, let the imputation model probabilities be 
\begin{eqnarray}
m_{i11} &=& p(t_i^* =1 \mid \boldsymbol{x}_i, y_{i}, t_i=1, r_i=1) \label{m11}\\
m_{i01} &=& p(t_i^* =1 \mid \boldsymbol{x}_i, y_{i}, t_i=0, r_i=1) \label{m01}\\
m_{i10} &=& p(t_i^* =0 \mid \boldsymbol{x}_i, y_{i}, t_i=1, r_i=1) \label{m10} \\
m_{i00} &=& p(t_i^* =0 \mid \boldsymbol{x}_i, y_{i}, t_i=0, r_i=1).\label{m00}
\end{eqnarray}
These probabilities are derived from an analyst's imputation model.  Although we include $t_i$ in the conditioning in \eqref{m11}--\eqref{m00} to facilitate notation,  
the analyst does not condition on $t_i$ when computing  
 $(m_{i11}, m_{i10}, m_{i01}, m_{i00})$ since $t_i$ is unknown when $r_i=1$. 

We can rearrange \eqref{eq2} with the true $e_i$ into terms comprising the units with $t_i=1$ and the units with $t_i=0$. For the terms comprising units with $t_i = 1$, the expectation over $p(t_i, r_i, t_i^{*} \mid \boldsymbol x_i, y_i)$ is 
\begin{eqnarray}
n^{-1}E\left(\sum_{i: r_i=0} \frac{y_{i1}t_i}{e_{i}} + \sum_{i: r_i=1, t_i=1, t_i^*=1} \frac{y_{i1}t_i}{e_{i}} - \sum_{i: r_i=1, t_i=1, t_i^*=0} \frac{y_{i1}t_i}{(1-e_{i})}\right)
&=& n^{-1}\sum_{i=1}^n y_{i1}\left((1-p_i) + p_i m_{i11} - \frac{p_{i}m_{i10} e_{i}}{(1-e_{i})}\right). \label{numyt}
\end{eqnarray}
For example, the first term in the summation arises because $E(y_i t_i/e_{i}) = (y_i/e_{i}) E(t_i\boldsymbol{1}_{r_i = 0}) = y_i (1-p_i)$ .  For the terms comprising units with $t_i = 0$, a similar computation reveals 
\begin{eqnarray}
n^{-1}E\left(\sum_{i: r_i=0} \frac{y_{i0}(1-t_i)}{(1-e_{i})} + \sum_{i: r_i=1, t_i=0, t_i^*=0} \frac{y_{i0}(1-t_i)}{(1-e_{i})} - \sum_{i: r_i=1, t_i=0, t_i^*=1} \frac{y_{i0}(1-t_i)}{e_{i}}\right)
&=& n^{-1}\sum_{i=1}^n y_{i0}\left((1-p_i) + p_im_{i00} - \frac{p_im_{i01}(1-e_{i})}{e_{i}}\right). \label{numyc}
\end{eqnarray}
Combining \eqref{numyt} and \eqref{numyc}, we see that 
\begin{equation}\label{eq:Eimp}
    E(\hat{\tau}^{(m)}) = \left(n^{-1}\sum_{i=1}^{n}y_{i1} - y_{i0}\right)  -n^{-1}\sum_{i=1}^n p_i y_{i1}\left(1 - m_{i11} + \frac{m_{i10} e_{i}}{(1-e_{i})}\right) - p_iy_{i0}\left(1-  m_{i00} + \frac{m_{i01}(1-e_{i})}{e_{i}}\right). 
\end{equation}
Thus, the bias of $\hat{\tau}^{(m)}$ is 
\begin{equation}
B = -n^{-1}\sum_{i=1}^n p_i y_{i1}\left(1 - m_{i11} + \frac{m_{i10} e_{i}}{(1-e_{i})}\right) - p_iy_{i0}\left(1-  m_{i00} + \frac{m_{i01}(1-e_{i})}{e_{i}}\right).\label{eq:B}
\end{equation}

The expression for $B$ in \eqref{eq:B} offers insight into the effects of imputation model specification.  In particular, we can get $B=0$ by having $m_{i11}=m_{00}=1$ so that $m_{10}=m_{01}=0$.
Of course, a perfect imputation model is not realistic in practice. On the other hand, consider using $e_{i}$ as the imputation model, that is, $m_{i01}=m_{i11}=e_{i}$ and $m_{i10}=m_{i00}=1-e_{i}$.  In this case, $B = -n^{-1}\sum_{i=1}^n p_i(y_{i1} - y_{i0})$, so that the bias potentially could be a substantial fraction of $\tau$.  As suggested by these two examples, the key is to find a model that drives $m_{i11}$ and $m_{i00}$ toward one.


\subsection{OMIT Imputation}\label{sec:OMIT}


Using (an estimate of) $e_i$ as the imputation model foregoes a potentially useful source of information, namely  the outcomes.
To illustrate, consider data  
where $y_{i1}\sim \mbox{Normal}(\tau + x_i\beta, 0.1) $ and $y_{i0} \sim \mbox{Normal}(x_i\beta, 0.1)$, with $\tau=40$ and $0 < x_i \beta < 10$ for all $x_i$.
It is clear that values of $y_{i}>40$ must have been assigned treatment and values of $y_{i}<10$ must have been assigned control, a fact we can use beneficially for imputation of $t_i^*$ for units missing $t_i.$  Indeed, if the distribution of $x_i$ is the same for the treated and control units, disregarding the information in the outcomes results in terribly biased inferences. 


Of course, this is an extreme and contrived example.
Nonetheless, the use of outcomes to assist the imputations of missing treatments is supported by the derivations in Section \ref{sec:biasMI}. In particular, 
the expressions for the imputation models in \eqref{m11}--\eqref{m00} explicitly condition on $y_i$.  Following Bayes rule, these probabilities are proportional to a product of a probability for the proposed treatment and a density for the observed $y_i$ given the proposed treatment, both given covariates:
\begin{equation}\label{impmodel}
    p(t_i^* = 1 \mid y_i, \boldsymbol{ x}_i)\propto p(t_i^* =1 \mid \boldsymbol x_i)f(y_i \mid  \boldsymbol{x}_i,t_i^* =1).
\end{equation}
One way to view this decomposition is that the predictive probability of treatment assignment 
can be sharpened, meaning that it is pushed closer to 0 or 1, by the information provided by the outcome model $f(y_i \mid \boldsymbol{x}_i,t_i^* =1)$. 

This construction provides considerable modeling flexibility: the analyst can choose the treatment assignment and outcome models based on the features of the data.  For the latter model, one insight from Proposition \ref{biased-CC} is that specification of $f(y_i \mid t_i, \boldsymbol x_i)$ in \eqref{impmodel} should capture interactions between the treatment indicators and covariates in the presence of treatment effect heterogeneity. For example, and previewing our data analyses, a linear regression with treatment-covariate interactions can be an effective, yet simple, outcome model to use in \eqref{impmodel}.  Both models  can leverage state-of-the-art statistical machine learning techniques such as Bayesian Additive Regression Trees (BART \citep{BART}) or Gaussian processes \citep{GP}. We explore these benefits in our simulations in Section \ref{sec:sim}. 

We also can show theoretically that 
outcome-assisted multiple imputation estimators can have desirable statistical properties, as stated in Theorem \ref{unbiased-imp}. Here, we find it helpful to switch from the finite population perspective of Section \ref{sec:biasMI} to a super-population perspective that presumes the potential outcomes are draws from distributions.  In particular, we 
are able to show that the expectation of the bias term in \eqref{eq:B} can equal zero.
\begin{theorem}\label{unbiased-imp}
 Suppose $y \mid \boldsymbol x, t  \sim f$ 
 with  $E(y)<\infty$. Suppose we impute missing $t_i$ using
$p(t^*_{i} =1 \mid \boldsymbol x_i, y_i) \propto e_i f(y_i \mid \boldsymbol x_i,t^*_{i} = 1),$
    where $e_i$ is the true propensity score given covariate profile $\boldsymbol x_i$.
    Then $E(B) = 0$, where $B$ is defined as in \eqref{eq:B} and the expectation is taken with respect to the density $f$.
\end{theorem}

\begin{proof}
Define $d_{i11} = f(y_{i1} \mid \boldsymbol{x}_i, t_i^*=1)$ to be the density of $y_{i1}$ under the outcome model when $t_i^{*} = t_i=1$.  Similarly, let $d_{i00} = f(y_{i0} \mid \boldsymbol{x}_i, t_i^*=0)$ be the density of $y_{i0}$ under the outcome model when $t_i^{*} = t_i =0$.  Likewise, we define 
$d_{i10}$ and $d_{i01}$ to  be the densities under the outcome model of $y_{i1}$ when $t_{i}^* =0$ (and $t_i =1$)  and of $y_{i0}$ when $t_{i}^{*} = 1$ (and $t_i = 0$), respectively.   
 
Under the assumptions of Theorem \ref{unbiased-imp}, 
we can write \eqref{m11}--\eqref{m00} as
\begin{eqnarray}
m_{i11} &=& d_{i11}e_{i}/(d_{i11}e_{i} + d_{i10}(1-e_{i})) = C_{i1}d_{i11}e_{i}\label{m11d}\\ 
m_{i01} &=& d_{i01}e_{i}/(d_{i01}e_{i} + d_{i00}(1-e_{i})) = C_{i0}d_{i01}e_{i}\label{m01d}\\
m_{i10} &=& d_{i10}(1-e_{i})/(d_{i11}e_{i} + d_{i10}(1-e_{i}))= C_{i1}d_{i10}(1-e_{i})\label{m10d}\\ 
m_{i00} &=& d_{i00}(1-e_{i})/( d_{i10}e_{i} + d_{i00}(1-e_{i}) ) = C_{i0}d_{i00}(1-e_{i}).\label{m00d}
\end{eqnarray}

Using \eqref{m11d}--\eqref{m00d}, \eqref{numyt} becomes
\begin{align}
\sum_{i=1}^n y_{i1}\biggl((1-p_i) + p_i C_{i1}d_{i11}e_{i} - p_{i}C_{i1}d_{i10}e_{i})\biggr) \label{expand}  
&= \sum_{i=1}^n y_{i1}\biggl((1-p_{i}) + p_iC_{i1}(d_{i11}e_{i} - d_{i10}e_{i}  + d_{i10} - d_{i10})\biggl)\\
&=\sum_{i=1}^{n} y_{i1}(1 - p_{i} C_{i1} d_{i10}).
\end{align}
By adding and subtracting $d_{i01}$ within the summation for the analogue to \eqref{expand}, \eqref{numyc} reduces to 
$n^{-1}\sum_{i=1}^{n} y_{i0}(1 - p_{i} C_{i0} d_{i01})$.  
Substituting each of these terms into \eqref{eq:Eimp}, we see that the bias term becomes
\begin{equation}
B=  - n^{-1} \sum_{i=1}^{n} y_{i1}p_iC_{i1}d_{i10i} - y_{i0}p_iC_{i0}d_{i01}.\label{expectationfinal}
\end{equation}

We leverage the data generating outcome model to reason probabilistically about \eqref{expectationfinal}. Specifically, we compute its expectation with respect to the conditional distribution of $y \mid \boldsymbol x, t$. 
{Since we assume the analyst uses the correct outcome model in the imputations,} for any $i$, we have $d_{i01} = f(y_{i0} \mid \boldsymbol x_i, t_i=1)$ and $d_{i10} = f(y_{i1} \mid \boldsymbol x_i, t_i=0)$. Thus, we have 
\begin{align}
 E(y_{i1}p_iC_{i1}d_{i10}) &=  p_i\int \frac{y f(y \mid \boldsymbol x_i, t_i =0)}{(e_{i} f(y \mid \boldsymbol{x}_i, t_i = 1) + (1-e_{i}) f(y \mid \boldsymbol x_i,t_i =0)} f(y \mid \boldsymbol x_i, t_i =1) dy \label{exp1}\\
    E(y_{i0}p_iC_{i0}d_{i01}) &= p_i\int \frac{y f(y \mid\boldsymbol x_i, t_i =1)}{(e_{i} f(y \mid \boldsymbol x_i, t_i =1) + (1-e_{i}) f(y \mid \boldsymbol x_i, t_i =0)} f(y \mid \boldsymbol x_i, t_i =0) dy. \label{exp0}
\end{align}
Comparing \eqref{exp0} and \eqref{exp1}, we see that both integrals match exactly.  Therefore, for any $D^{(m)}$, we have $E(B) = 0$ when imputations are made using the true outcome and propensity score models (for MAR treatments).  
\end{proof}

Of course, in practice we estimate both the outcome and propensity score models imperfectly. Nonetheless, the fact that one can achieve unbiasedness using \eqref{impmodel} demonstrates its potential advantages within a multiple imputation procedure.


\subsection{Implementation Guidelines}\label{sec:implementation}

We describe the general OMIT algorithm below. Here,  
we write $\hat{f}(y_i \mid \boldsymbol x_i, t^{*}_i)$ to represent  the analyst's estimated density for  $y_i$ given $\boldsymbol x_i$ and $t_{i}^{*}$ and $\hat{e}_i$ to represent the analyst's estimated propensity score.
 \begin{enumerate}
    \item \textbf{Model Estimation:} Using the complete data $\{(y_i, x_i, t_i): r_i = 0\}$, estimate the parameters of the models needed for computing $\hat{f}(y_i \mid t^{*}_i, \boldsymbol{x}_i)$  and $\hat{e}_i$.
    \item \textbf{Imputation:} For each unit $i$ where $r_i = 1$,
        \begin{enumerate}
            \item Compute unnormalized probabilities 
            $\tilde{q}_{i1} = \hat{e}_i\hat{f}(y_i \mid \boldsymbol x_i, t^{*}_i = 1)$ and   $\tilde{q}_{i0} = (1-\hat{e}_i)\hat{f}(y_i \mid \boldsymbol x_i, t^{*}_i = 0)$.
        \item Compute 
            $\hat{q}_{i} = \tilde{q}_{i1}/(\tilde{q}_{i0} + \tilde{q}_{i1}).$
        \item For $m = 1,\dots, M$,
                 impute $t_i^*$ for use in $D^{(m)}$ by sampling independently from a   
                $\mbox{Bernoulli}(\hat{q}_{i})$.
        \end{enumerate}
    \end{enumerate}

With no missingness among the covariates and outcomes, and MAR treatments, analysts can fit their outcome and propensity score models to the complete case data and still obtain unbiased estimates for the model parameters \cite{rubin1976inference}; they need not utilize the data for units with $r_i=1$ for those models. For further computational convenience, analysts can use point estimates of model parameters to compute \eqref{impmodel}. The simulation results in Section \ref{sec:sim} suggest that this approach can perform well provided there is enough observed data for standard regression modeling exercises.

To set the propensity score and outcome models, analysts should utilize standard diagnostic tools to aid model specification,
such as residual diagnostics.
We emphasize that the outcome modeling is used only for imputation.
Any misspecification of that model affects only the fraction of observations with missing treatments, and even then only those units where $t_i^* \neq t_i$ contribute to potential biases. For example, if the data exhibit positive correlation of $y$ and some covariate $x$, using a linear outcome model for OMIT may capture that correlation sufficiently to encourage imputing $t_i^*=1$ for records with large $x_i$, even if the true relationship between $y$ and $x$ is nonlinear.


When covariates or outcomes have missing values, analysts can use off-the-shelf routines like multiple imputation by chained equations (MICE) for the imputation, being sure to include both the covariates and outcomes as predictors in the conditional model for $t$.  This imputes $t_i^*$ from a single model  $f(t_i \mid \boldsymbol{x}_i, y_i)$, which is not the same as \eqref{impmodel}; see Section \ref{sec:conclusion} for further discussion of this point.  Alternatively, analysts could adopt a two-stage approach.  First, apply a standard MICE algorithm to impute values for the missing covariates, outcomes, and treatments.  Then, in each $D^{(m)}$, blank the imputations of the treatments and re-impute $t_i^*$ for units with $r_i=1$ using \eqref{impmodel}.  Here, we estimate  \eqref{impmodel} in each $D^{(m)}$ using its completed data  $\{\boldsymbol{x}_i, y_i, t_i: r_i=0\}$.  We leave investigation of this approach to future research.




\section{Simulation}\label{sec:sim}

In this section, we examine the benefits of OMIT using simulation studies. 
As part of the simulations, we examine the robustness of the OMIT algorithm to misspecification of the outcome model.
For the data generating distribution, we use 
\begin{align}
    \boldsymbol x &\sim \mbox{MVN}_{10}(\boldsymbol 0, \rho = 0.4)\label{xmodel} \\
     t \mid \boldsymbol x &\sim \mbox{Bernoulli}\{\Phi(\alpha_t + \boldsymbol x'\boldsymbol \beta_t)\} \label{tmodel}\\
     y \mid \boldsymbol x, t = 1 & \sim N(2 + x_1 + \beta_{y}x_{2}^{2}, \sigma^{2})\label{y1model}\\
    y \mid \boldsymbol x, t = 0 & \sim N(1 + x_1 + x_{2}^{2}, \sigma^{2})\label{y0model}\\
    r \mid \boldsymbol x,y &\sim \mbox{Bernoulli}\{\Phi(\alpha_{r} + \boldsymbol x' \boldsymbol\beta_r + \gamma \tilde{y})\}.\label{rmodel}
\end{align}
We sample from these distributions to construct a finite population comprising $n=1000$ observations, $\{(y_i, \boldsymbol x_i, r_i, t_i)\}_{i=1}^{n}$ . Specifically, we first  generate $1000$ independent covariate vectors from \eqref{xmodel}. For each $\boldsymbol x_i$, we simulate $y_{i1}$ from \eqref{y1model} and $y_{i0}$ from \eqref{y0model}. We fix the covariates and potential outcomes so that the ATE is $\tau = 1000^{-1} \sum_{i=1}^{1000} y_{i1} - y_{i0}$. We then  simulate 500 instances of random treatment assignments from \eqref{tmodel}. In each of these iterates, we use the realized $t_i$ to create the observed outcomes $y_i = t_i y_{i1} + (1-t_i)y_{i0}$ where $i=1, \dots, 1000$. Finally, in each iterate, we simulate each unit's $r_i$ from \eqref{rmodel} and remove the realized $t_i$ whenever $r_i =1$.

In \eqref{tmodel}, $\alpha_{t} = \Phi^{-1}(0.4)$, where $\Phi$ is the standard normal distribution function, resulting in approximately 40\% treated individuals. All elements in the $10$-dimensional $\boldsymbol\beta_{t}$ equal zero except for the second component $\beta_{t2} = 0.35$. Thus, individuals with larger values of $x_{2}$ are more likely to be treated. The outcome 
distributions in \eqref{y1model} and \eqref{y0model}  
encode an individual treatment effect of $\tau_i = 1+(\beta_y-1)x_{i2}^2$ with an average treatment effect of $1+(\beta_y-1)$. 
When $\beta_y=1$, we have a homogeneous treatment effect. By increasing $\beta_{y}$, we increase the heterogeneity of the individual treatment effects $\tau_i$ and their average. We consider three values of $\beta_{y} \in \{1,4,7\}$. We also also vary $\sigma \in \{1,2,4\}$ to increase the amount of noise in the outcome distribution. 

Since the missingness mechanism is conditionally independent of the treatments  given the covariates and outcomes, \eqref{rmodel} is a missing at random mechanism. 
We consider three values of $\alpha_{r} \in \{\Phi^{-1}(0.1), \Phi^{-1}(0.3), \Phi^{-1}(0.5)\}$ corresponding to marginal proportions of missing treatments of roughly 10\%, 30\%, and 50\%, respectively. We set all components of $\boldsymbol \beta_r$ to zero besides $\beta_{r2} = 0.75$. In addition, $\gamma = .1$ and $\tilde{y}$ is a mean-zero, unit-variance version of the realized outcomes. Thus, the covariate $x_{2}$ is centrally important to the study. It moderates the effects of the treatment and contributes to the prediction of the treatment assignment and missingness. 


In each of the 500 iterations,
we perform multiple imputation via the OMIT algorithm, varying the specification of the outcome model. As one option, we employ linear regression with main effects and first-order interactions between covariates and treatment status; we refer to this as OMIT:lm. For OMIT:lm, we include all covariates and their first-order interactions with treatment status in the model, without employing a mechanism for shrinkage or variable selection. As another option, we utilize Bayesian Additive Regression Trees (BART \citep{BART}), which can capture both nonlinearity in the response surface and interactions between the covariates and treatment status; we refer to this as OMIT:BART.  Further details on model specification including codes in  \verb|R| are available in the supplement. Finally, we conduct OMIT using the specification of the outcome models from \eqref{y1model}-\eqref{y0model}, which amounts to a linear regression of $y$ on $x_{1}$, $x_{2}^{2}$, and an interaction between $x_{2}^{2}$ and $t$; we refer to this as OMIT:Correct. 
We use the point estimates of the parameters underlying each outcome model to compute $\hat{f}(y \mid \boldsymbol x, t)$ for imputation, estimated with  
complete case data.
For OMIT:BART, we take the average across $800$ posterior samples of the mean function and error variance to evaluate the outcome density.

For the propensity score model, 
we estimate a probit regression of $t$ on $x_{2}$, i.e., the treatment assignment mechanism is correctly specified, and use the coefficient estimates to calculate propensity scores $\hat{e}_i$. We use these propensity scores  in \eqref{impmodel} for each OMIT specification and for estimation of the ATE via \eqref{IPWmi} for all methods. As a comparison, we also use $\hat{e}_i$ alone to impute missing treatments; 
we refer to this procedure as ``Naive MI.''   Finally, we also compute results for a complete-cases analysis, which we refer to as CC.  


For each imputation method, we create $M=20$ completed data sets and compute point estimates and 95\% confidence intervals for the ATE using \eqref{IPWmi}. We estimate the 
variances in \eqref{eq:miubar} using the \verb|survey| package in \verb|R| \citep{lumley2020package}, specifying an unequal probability sampling design based on the estimated propensity scores.  For evaluating the repeated sampling performances of the methods, in each iterate $rep=1, \dots, 500$, let $\hat{\tau}_{rep}$ be the point estimator for some particular method.  We consider
 the distributions of the differences $\hat{\tau}_{rep} - \tau$ across simulation iterations to get a sense of the biases. We also compute 
 the empirical coverage rates of the 95\% confidence intervals.

\begin{figure}[t]
    \centering
\includegraphics[width=0.9\linewidth,keepaspectratio]{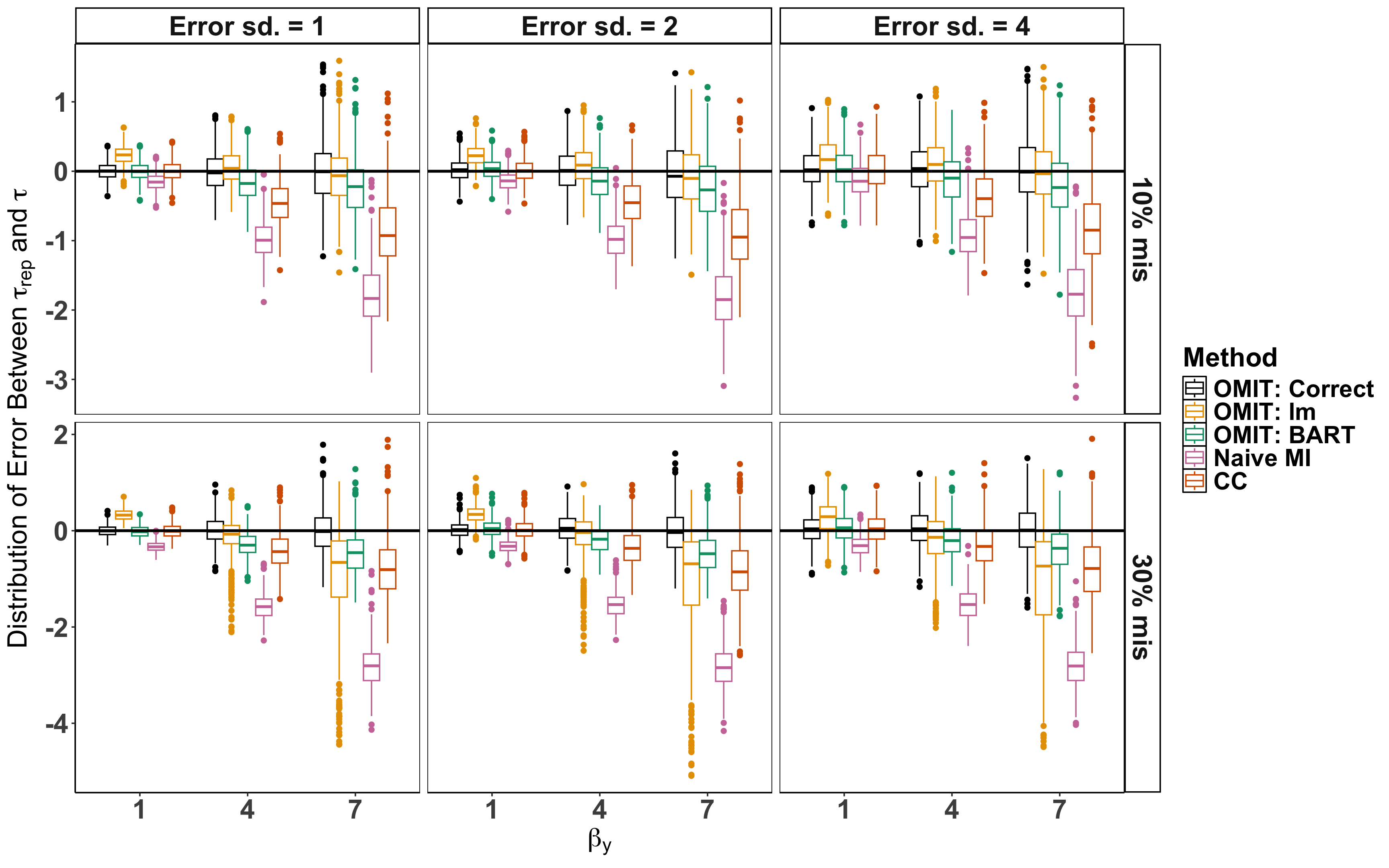}
    \caption{The distribution of differences between $\hat{\tau}_{rep}$ and $\tau$ across methods and simulation settings.} 
    \label{fig:bias}
    \end{figure}

Figure \ref{fig:bias} 
and Figure \ref{fig:coverage} display the distributions of the differences and the empirical coverage rates for the three versions of OMIT, Naive MI, and CC analysis. We include results for the 10\% and 30\% missingness scenarios for each combination of $(\beta_y, \sigma)$. Results for the 50\% missingness settings are qualitatively similar and are available in the supplement.

The simulation results accord with the theory in Section \ref{sec:cc} and Section \ref{sec:theory}.  CC analysis is approximately unbiased only when treatment effects are homogeneous, i.e., $\beta_y=1$, per Proposition \ref{biased-CC}.  Correspondingly, CC analysis tends to have lower than nominal coverage rates in all scenarios except when $\beta_y=1$. 
Naive MI has coverage rates that are generally below 50\%, especially in the simulation settings with heterogeneous treatment effects, and its mean-squared errors are two to ten times larger than those for the other imputation methods.  This poor performance occurs despite using a correctly specified propensity score model.
In contrast, under OMIT:Correct, 
multiple imputation inferences are approximately unbiased, as suggested by Theorem \ref{unbiased-imp}, with at least the nominal 95\% coverage rate.  It dominates Naive MI and CC analysis.  


Turning to comparisons among the three OMIT implementations, not surprisingly OMIT:Correct has the best performance. However, the other two OMIT models still perform reasonably well. OMIT:lm displays low bias and high nominal coverage rates when there is moderate treatment effect heterogeneity, more noise in the outcomes, and less missingness. In these scenarios, the quadratic relationship between $x_{2}$ and $y$ is less prominent, and the linear outcome model offers a reasonable fit. As the missingness proportion and $\beta_{y}$ increase, OMIT:lm performance begins to deteriorate, as the response surface becomes more sharply nonlinear and imputed treatments comprise a larger portion of the data. In these settings, OMIT:BART still offers close to nominal coverage rates and lower biases, demonstrating the benefits of leveraging nonparametric regression models that are capable of adapting to such features in the data. Further, because the linear outcome model in OMIT:lm does not provide any mechanism for shrinkage or selection, non-zero regression coefficients for the treatment-covariate interactions degrade its performance when the treatment effect is homogeneous. 

OMIT:BART generally outperforms CC and Naive MI, lending further support to the advantages of OMIT with flexible modeling.  OMIT:lm outperforms Naive MI except when $\beta_y=1$, particularly with less missing data and when $\sigma$ is large enough that a linear outcome model offers a reasonable approximation to the quadratic response surface.  OMIT:lm tends to perform better or nearly as well as the CC analyses, again except in the case of homogeneous treatment effects and strong nonlinearity.  Overall, 
the performances of the OMIT procedures demonstrate the benefits of accounting for the outcomes within the imputation models.

\begin{figure}[t]
    \centering
\includegraphics[width=0.9\linewidth,keepaspectratio]{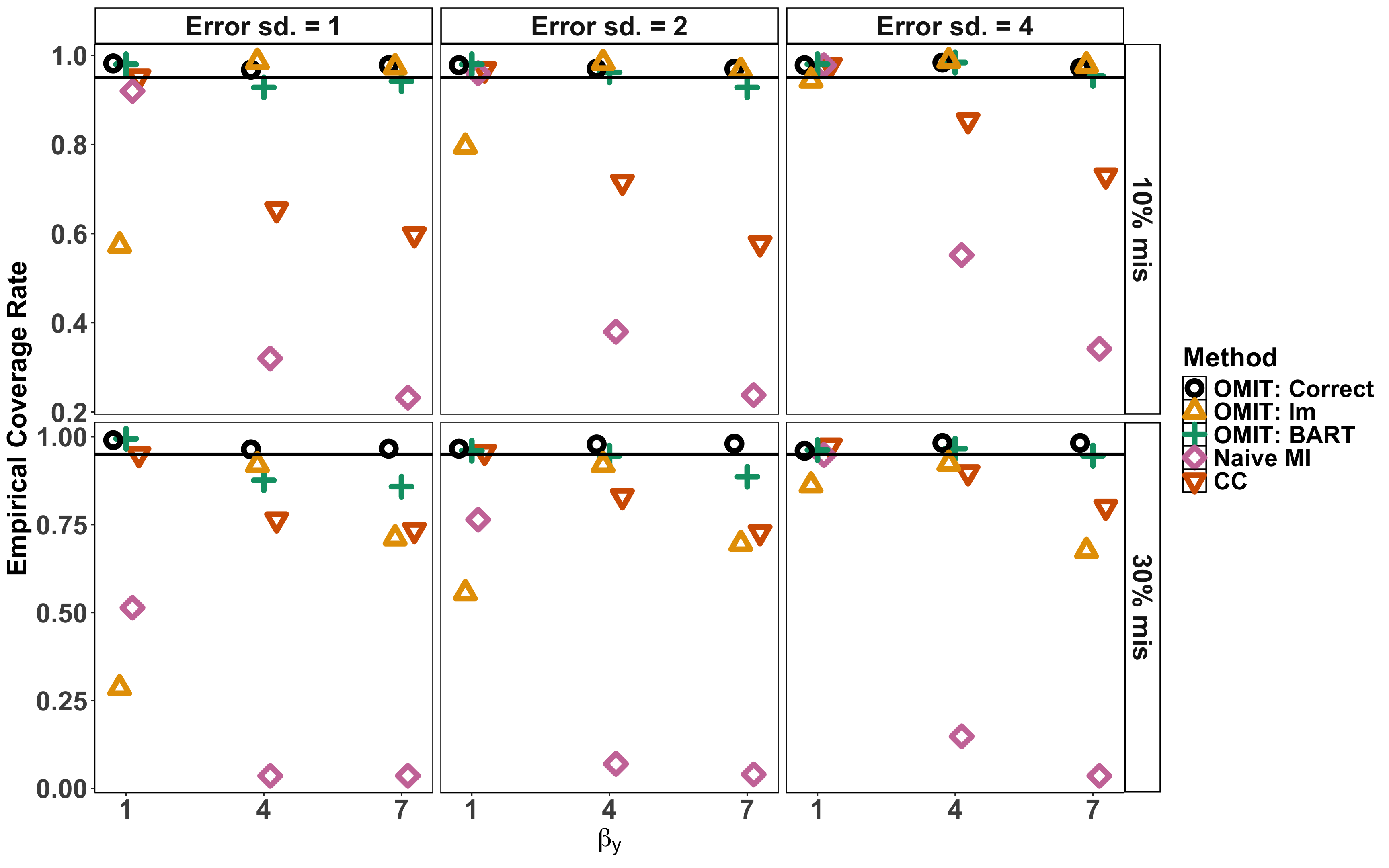}
    \caption{Empirical coverage rates of 95\% confidence intervals for the different methods across simulation scenarios.}
    \label{fig:coverage}
\end{figure}

\section{Application with the Longitudinal Survey of Youth}\label{sec:NLSY}

We now illustrate OMIT using a subset of variables from the National Longitudinal Survey of Youth (NLSY79 \cite{nysl}), as  analyzed in  \cite{mitra2016comparison} and  \cite{mitra2023latent}.  The latter work estimates the average effect of family income on Peabody Individual Assessment Test math scores (PIATM \cite{PIATM}) for the children in the study. This analysis does not claim definitive conclusions about a causal effect of income; rather, 
it uses the analysis to demonstrate the sensitivity of the estimand to different missing data approaches.
We also take 
this perspective.

We consider a subset of covariates comprising the child's sex, race, and birth weight; the mother's intelligence as measured by an armed forces qualification test; the mother's age at birth; and, indicators of whether the mother had any college-level education, whether the child was breastfed for over 24 weeks, whether the child attended daycare, and whether the child was born prematurely. The response is PIATM.  We scale all numeric covariates to mean zero and unit variance. Each child's treatment status is $t_i=1$ if 
their family income at the time of birth was greater than the empirical 75th percentile of income and $t_i=0$ if not. 
For purposes of illustration, we take as the study population all children with completely observed covariates and outcomes, but potentially missing treatments, which yields $n = 2189$ observations.
Approximately 20\% have missing treatment status.

We assume that treatment status is missing at random. Distributional summaries of several covariates, including race, mother's age, mother's intelligence, and PIATM scores, differ depending on whether income status is observed or missing. We observe an increase of nearly 8 points in average PIATM scores 
between the treated and control children, which provides initial evidence that family income has a positive effect on the outcome. Children with missing treatment status have covariate profiles that look more like the profiles of children with observed $t_i=0$.  For example, their mothers are younger with lower levels of education on average, and the children  are disproportionately non-white. Furthermore, these children have an average PIATM score of 99.7, which is close to the control group's average score of 99.5. By contrast, children in the treated group had an average PIATM score of 107. Taken together, we suspect the information in the  outcomes can sharpen imputation of missing treatment status.

We perform OMIT using two outcome models, namely an OMIT:BART and an OMIT:lm that uses a Gaussian linear model with main effects and first-order treatment-covariate interactions. Residual standard errors are slightly lower for the OMIT:BART model, demonstrating a closer fit to the data. For the propensity scores, which we use for imputation of missing treatments and in the IPW estimators, we estimate a probit regression with main effects for all covariates. We also examine imputations where we use only these propensity scores for the imputation model, i.e., we fit Naive MI. For each method, we generate $M=20$ completed datasets, which we use to generate IPW point estimates and standard errors as described in Section \ref{sec:MI} using the \verb|survey| package in \verb|R|. We also compare these estimates and standard errors to a complete case (CC) analysis that drops any observation missing treatment status.

\begin{table}[t]\centering
\caption{Average treatment effect estimates of family income status on PIATM scores in the illustrative NLSY analysis. 
}\label{tab1}
\begin{tabular}{@{}lcccc@{}}
\toprule
\textbf{Imputation Method} & \textbf{OMIT:lm} & \textbf{OMIT:BART} & \textbf{Naive MI} & \textbf{CC}  \\ \midrule
\textbf{ATE Estimate}      & 1.64        & 2.08         & 0.71  & 0.80      \\
\textbf{(Std. Error)}  & (1.52)         & (1.42)         & (1.46)  & (1.51)       \\[0.5em]
\bottomrule
\end{tabular}
\end{table}

Table \ref{tab1} summarizes the results. The 
ATE estimates under the OMIT models suggest stronger treatment effect   estimates than both Naive MI and CC analysis. To see why this is the case, in Figure \ref{probcomp} we compare 
the predictive probabilities of positive treatment status ($t_i^* = 1$) for children missing $t_i$ under each OMIT imputation model to those for Naive MI. 
We restrict comparisons to observations with outcome-assisted predictive probabilities greater than 0.25, as this represents the empirical rate of treatment in the observed data. Among these observations,
62\% of OMIT:BART predictive probabilities and 64\% of OMIT:lm predictive probabilities exceed the corresponding Naive MI predictive probabilities. Thus, the OMIT approaches are more likely to impute positive treatment status than the Naive MI procedure.  These individuals also have an average PIATM score over 108, which is similar to the average for the observed treated children.
Evidently, the outcome models have captured a positive association between treatment status and PIATM, which they utilize to sharpen the predictive probabilities for imputation.

\begin{figure}
    \centering
    \includegraphics[width=0.49\linewidth]{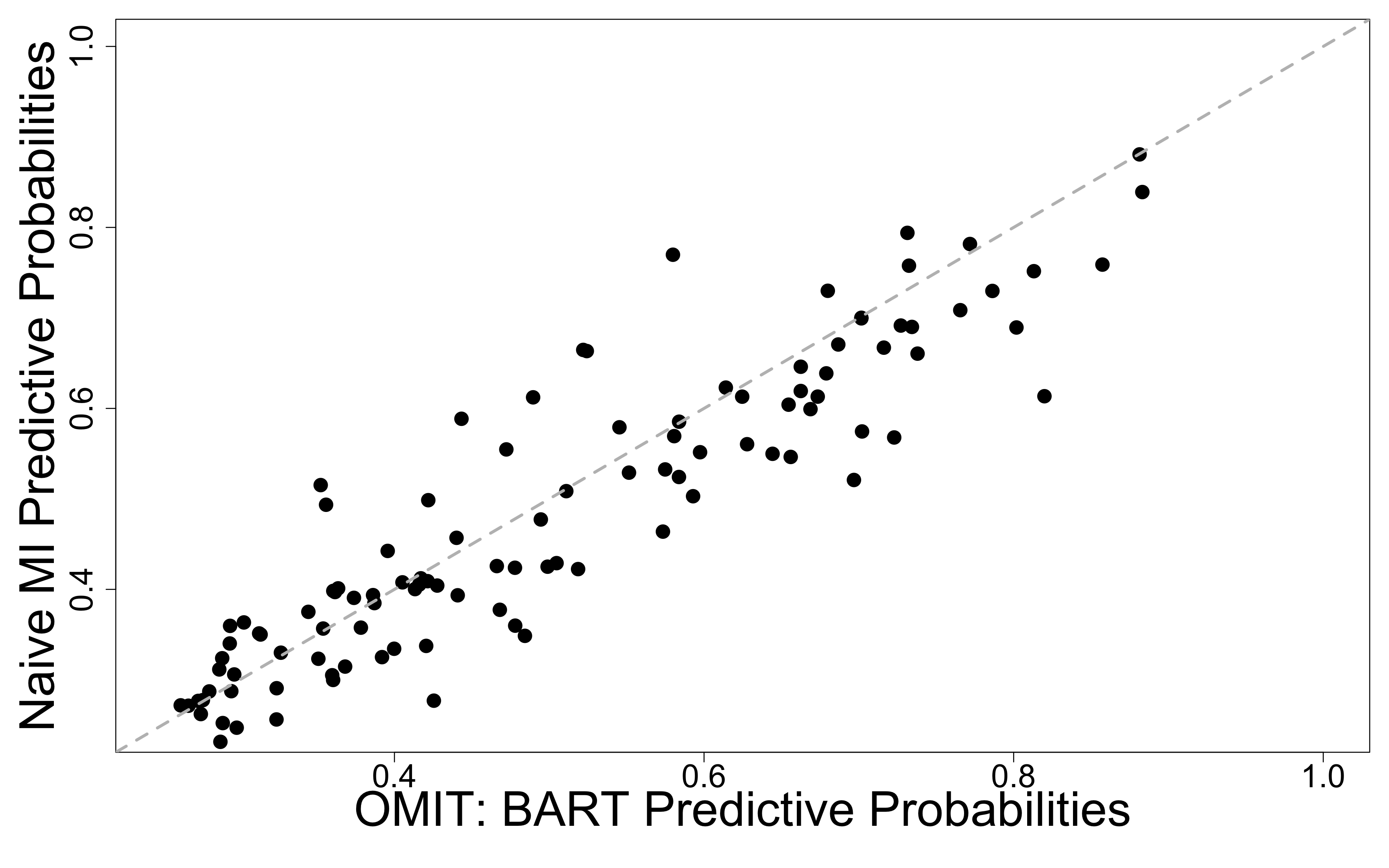}
    \includegraphics[width=0.49\linewidth]{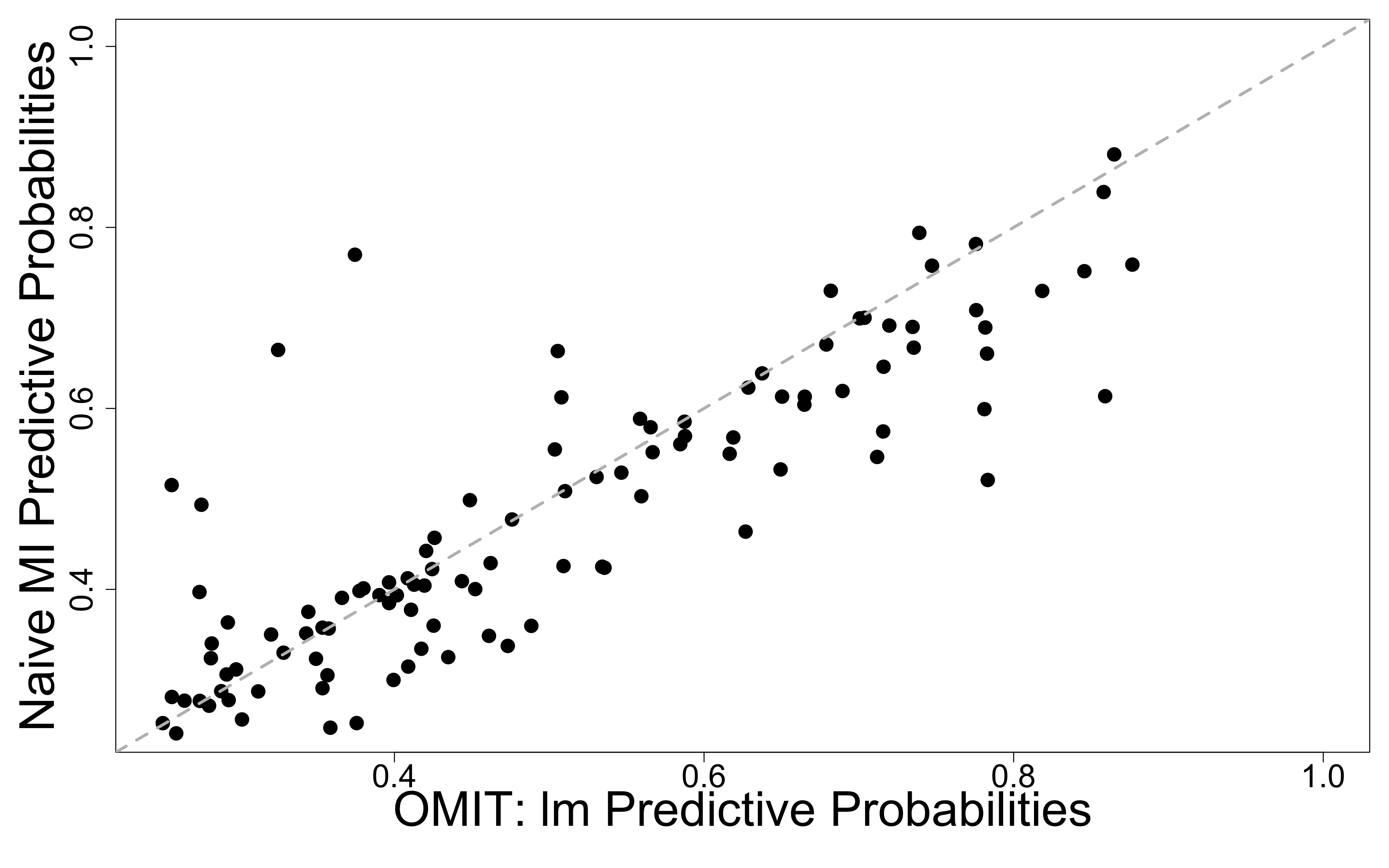}
    \caption{Comparison of predictive probabilities from the two OMIT models to those from Naive MI. } 
    \label{probcomp}
\end{figure}


\section{Conclusion}\label{sec:conclusion}

We find that adhering to a design-first philosophy  
with missing treatments, that is, imputing missing treatments using $p(t = 1 \mid \boldsymbol x)$ without using $y$, does not allow valid inference in general.  It seems that not strictly following this philosophy is a price to pay for having missing data \cite{heck:reiter, kamat:reiter, guha:reiter:merc}.  We also find some robustness to modest misspecifications of the outcome model.  Essentially, as long as the outcome model appropriately sharpens the predictive probabilities from imputation, the outcome-assisted procedure has the potential to mitigate bias in multiple imputation estimates of the ATE.  Of course, the outcome-assisted model is not robust to gross misspecification of the outcome model. Indeed, the use of propensity scores only for imputation---which performed relatively poorly in the simulation studies---can be viewed as an outcome-assisted model with a  uniform distribution for the outcomes. Thus, analysts should strive to use outcome models that reasonably describe the data as determined from careful model fitting procedures, or rely on flexible models like BART.

As mentioned briefly in Section \ref{sec:implementation}, analysts could use a single model  $p(t_i|\boldsymbol{x}_i, y_i)$ in lieu of the two sub-models in OMIT, even when $(\boldsymbol{x}_i, y_i)$ are fully observed. However, we speculate that analysts may find it more natural and easier to specify models for the propensity scores and outcomes, particularly when treatment effects are heterogeneous. 
The supplement describes a simulation study where the response surface for the outcomes is  highly nonlinear and the analyst uses a main effects only specification of the linear predictor for the model for $t$, as might occur in a bespoke application of MICE.  The misspecification results in badly biased estimates of the ATE, whereas OMIT:lm and OMIT:BART offer more accurate inferences. 
Regardless of the approach, we recommend using the outcomes to assist in multiple imputation of missing treatments.  

Our work only covers scenarios with missing at random treatments.  With not missing not at random treatments,
the imputation models and general OMIT algorithm need to be revised, since (i) we require a model for the missingness and (ii) the complete case data are not sufficient for consistent estimation of the outcome model or treatment assignment mechanism. 

\bmsection*{Author Biography}

\bibliography{wileyNJD-AMA.bib}

\end{document}


\title{Supplement to \\ ``Outcome-Assisted Multiple Imputation of Missing Treatments''}
\date{}

\author{Joseph Feldman and Jerome P. Reiter}

\ifblinded
\author{}
\else

\author{Joseph Feldman and Jerome P. Reiter}

\fi

\maketitle
\large 

\vspace{-10mm}

In this supplementary material, we provide a proof of Proposition 1 in the main text.  We also provide code for fitting the OMIT models using a linear outcome model and using BART.  Finally, we provide additional simulation results for a scenario where imputing directly from $p(t_i \mid \boldsymbol x_i, y_i)$ results in worse inferences compared to OMIT.  Here, equation numbers refer to the main text unless preceded by the letter A, for example, (A.1).

\appendix

\section{Proof of Proposition 1}
\begin{proposition}\label{biased-CC} Assume strong ignorability and SUTVA, and that $p(\sum_{i=1}^{n} \boldsymbol{1}(r_i = 1) = n) = 0$ for any $n$. Assume we use the true propensity scores in $\hat{\tau}_{CC}$ as defined in \eqref{IPW-CC}. 
 If one of the following conditions holds:
\begin{itemize}
    \item there are homogeneous conditional treatment effects, i.e., $\tau_{i} = \tau_{j}, \forall i \neq j \in \{1, \dots, n\}$,
    \item treatments are missing completely at random (MCAR), i.e., 
    $p(r_{i} = 1 \mid \boldsymbol x_i) = p$ \ for \ $i=1, \dots, n$,
\end{itemize}
then $E(\hat{\tau}_{CC}) = \tau$, where the expectation is taken with respect to the joint distribution of the treatment assignment and missingness mechanism.
\end{proposition}


\begin{proof}
    We prove the implications of each condition one at a time.  We begin by presuming homogeneous treatment effects.  As a reminder, the complete-cases IPW estimator in \eqref{IPW-CC} is
    \begin{eqnarray}
    \hat{\tau}_{CC} &=& n_{obs}^{-1}\sum_{i: r_{i} = 0} \frac{t_i y_i}{e_i}  - \frac{(1-t_i)y_i}{(1-e_i)}\\
        &=& n_{obs}^{-1}\sum_{i=1}^{n} \mathbf{1}_{r_i = 0} \left(\frac{t_i y_{i1}}{e_i}  - \frac{(1-t_i)y_{i0}}{(1-e_i)}\right)
    \end{eqnarray}
Under the causal assumptions, we can use an iterated expectation over the distributions of the assignment mechanism for $t_1, \dots, t_n$ and the nonresponse mechanism for $r_1, \dots, r_n$. Here, we treat $\{\boldsymbol x_i, y_i\}_{i=1}^n$ as fixed. We have 
\begin{eqnarray}
E(\hat{\tau}_{CC})
 &=&      
E\left(E\left(n_{obs}^{-1}\sum_{i=1}^{n}\mathbf{1}_{r_i = 0} \left(\frac{t_i y_{i1}}{e_i}  - \frac{(1-t_i)y_{i0}}{(1-e_i)}\right) \mid r_1, \dots, r_n\right)\right) \nonumber\\ 
        &=&  E \left(n_{obs}^{-1} \sum_{i=1}^{n} \mathbf{1}_{r_i=0} \tau\right) = E(n_{obs}^{-1} n_{obs} \tau) = \tau,\label{simp2}
    \end{eqnarray}
    since $y_{i1} - y_{i0} = \tau$ for any $i$ by assumption.
    
    
    Now suppose instead that treatments are missing completely at random. Then, we have 
\begin{eqnarray}
E(\hat{\tau}_{CC})
 &=&      
E\left(E\left(n_{obs}^{-1}\sum_{i=1}^{n}\mathbf{1}_{r_i = 0} \left(\frac{t_i y_{i1}}{e_i}  - \frac{(1-t_i)y_{i0}}{(1-e_i)}\right) \mid r_1, \dots, r_n\right)\right) \nonumber\\ 
        &=&  E \left(n_{obs}^{-1} \sum_{i=1}^{n} \mathbf{1}_{r_i=0} \tau_i\right)\\
&=& E\left(E\left(n_{obs}^{-1} \sum_{i=1}^{n} \mathbf{1}_{r_i=0} \tau_i \mid n_{obs}\right)\right) \label{simp2a}\\
        &=& E\left(n_{obs}^{-1} \sum_{i=1}^n (n_{obs}/n) \tau_i\right) = \tau.
    \end{eqnarray}
Here, $E\left(\mathbf{1}_{r_i=0} \mid n_{obs}\right) = n_{obs}/n$ for $i=1, \dots, n$ because of the MCAR mechanism, which selects a simple random sample of $n_{obs}$ out of $n$ units' treatments to be observed. 


\end{proof}

\section{Code for Outcome Assisted Multiple Imputation Using Linear Regression}\label{OMIT:lm}
\begin{lstlisting}[style=Rstyle]
### Outcome-assisted multiple imputation using linear regression

#observed data
## y: n realized outcomes
## r: n missingness indicators
## x: n \times p covariate matrix
## t: n treatment assignment indicators, with t_i = NA when r_i = 1
data<- data.frame(y,r,x,t)

#get propensity scores for each observation using main effects linear regression fit to complete cases
e_x<-predict(glm(t~x,data = data,family = binomial(link = "probit")), newdata = data,type = "response")

#fit outcome model to complete cases using main effects linear regression with treatment-covariate interactions to complete cases
outcome_model_lm<-lm(y~.*t, data = subset(data, select = c(y,x,t)))

#compute outcome assisted probabilities

## E[Y(1)| x]
y1_lm<- predict(outcome_model_lm, newdata = data.frame(cbind(subset(data , select = -c(y,t)),t = 1)))
## E[Y(0)| x]
y0_lm<-predict(outcome_model_lm, newdata = data.frame(cbind(subset(data , select = -c(y,t)),t = 0)))

#error standard deviation
sigma_lm = summary(y_mod_lm)$sigma

# compute unnormalized probabilites
tilde_q_1= e_x*dnorm(y,y1_lm, sd =sigma_lm) 
tilde_q_0=(1-e_x)*dnorm(y,y0_lm, sd =sigma_lm)

# extract sharpened probabilities
q_1<- tilde_q_1/(tilde_q_1 + tilde_q_0)

#number of multiple imputations
M = 20; imputed_data = vector('list',M)
for(m in 1:M){
  # mth completed data set
  data_m = data 
  # Impute
  data_m$t[data$r == 1] = rbinom(sum(data$r == 1),1,p_1)
  #save imputations
  imputed_data[[m]] = data_m
}
\end{lstlisting}
\section{Details on Outcome Assisted Imputation using BART}
We leverage the \verb|dbarts| packages under default settings to conduct OMIT:BART imputation. The excerpt below provides \verb|R| code for implementation, continuing the data analysis from Section \ref{OMIT:lm}.
\begin{lstlisting}[style=Rstyle]
#observed data
## y: n realized outcomes
## r: n missingness indicators
## x: n \times p covariate matrix
## t: n treatment assignment indicators, with t_i = NA when r_i = 1

library(tidyr)
library(dplyr)
library(dbarts)

data<- data.frame(y,r,x,t)

#get propensity scores for each observation using main effects linear regression fit to complete cases
e_x<-predict(glm(t~x,data = data,family = binomial(link = "probit")), newdata = data,type = "response")

#fit outcome model to complete cases using BART
X_bart<-dbarts::makeModelMatrixFromDataFrame(subset(data %>% drop_na(), select = -c(y,r)) 
outcome_model_bart<- dbarts::bart(X_bart, data$y[r==0],keeptrees = TRUE, verbose = F)


#compute outcome assisted probabilities

## E[Y(1)| x]
y1_bart<-colMeans(predict(y_mod_bart,makeModelMatrixFromDataFrame(data.frame(subset(dat_obs, select = -c(y, t)),t=1), drop = F)))
## E[Y(0)| x]
y0_bart<-colMeans(predict(y_mod_bart,makeModelMatrixFromDataFrame(data.frame(subset(dat_obs, select = -c(y, t)),t=0), drop = F)))

#error standard deviation
sigma_bart<-y_mod_bart$sigest

# compute unnormalized probabilites
tilde_q_1= e_x*dnorm(y,y1_bart, sd =sigma_bart) 
tilde_q_0=(1-e_x)*dnorm(y,y0_lm, sd =sigma_bart)

# extract sharpened probabilities
q_1<- tilde_q_1/(tilde_q_1 + tilde_q_0)

#number of multiple imputations
M = 20; imputed_data = vector('list',M)
for(m in 1:M){
  # mth completed data set
  data_m = data 
  # Impute
  data_m$t[data$r == 1] = rbinom(sum(data$r == 1),1,p_1)
  #save imputations
  imputed_data[[m]] = data_m
}
\end{lstlisting}

\section{Additional Simulation Results}
\subsection{Simulations from main text with 50\% missing treatments}
In this section, we present additional simulation results using the design features described in the main text.  First, we include results when 50\% of the treatments are missing at random. Second, we include results where we impute missing treatments using one model for $p(t_i \mid \boldsymbol x_i, y_i)$, which is specified as a probit regression using main effects of all covariates and the (standardized) outcome.
We label this approach Naive+Y MI. The results are summarized in Figure  \ref{fig:biasv2} and \ref{fig:coveragev2}. For convenience, these figures report the results presented in main text. 

Even with  50\% missingness, once again using OMIT with the correct outcome model results in approximately unbiased inferences with at least 95\% empirical coverage rates. With 50\% missingness, OMIT:BART generally has better performance than OMIT:lm. In these settings, there is weak signal in the observed data to estimate the parameters in the linear regression. 
We note that standard residual diagnostics confirm this feature.  Interestingly, the CC estimates in the 50\% missing data setting offer generally competitive inferences relative to the OMIT approaches. In this setting, the majority of observations in the data have similar probabilities of missing treatment, i.e., the missingness mechanism is close to MCAR. Thus, the performance of CC in these settings is supported by the insight from Proposition \ref{biased-CC}, especially since we have correctly specified the propensity score model.

\begin{figure}[t]
    \centering
\includegraphics[width=0.9\linewidth,keepaspectratio]{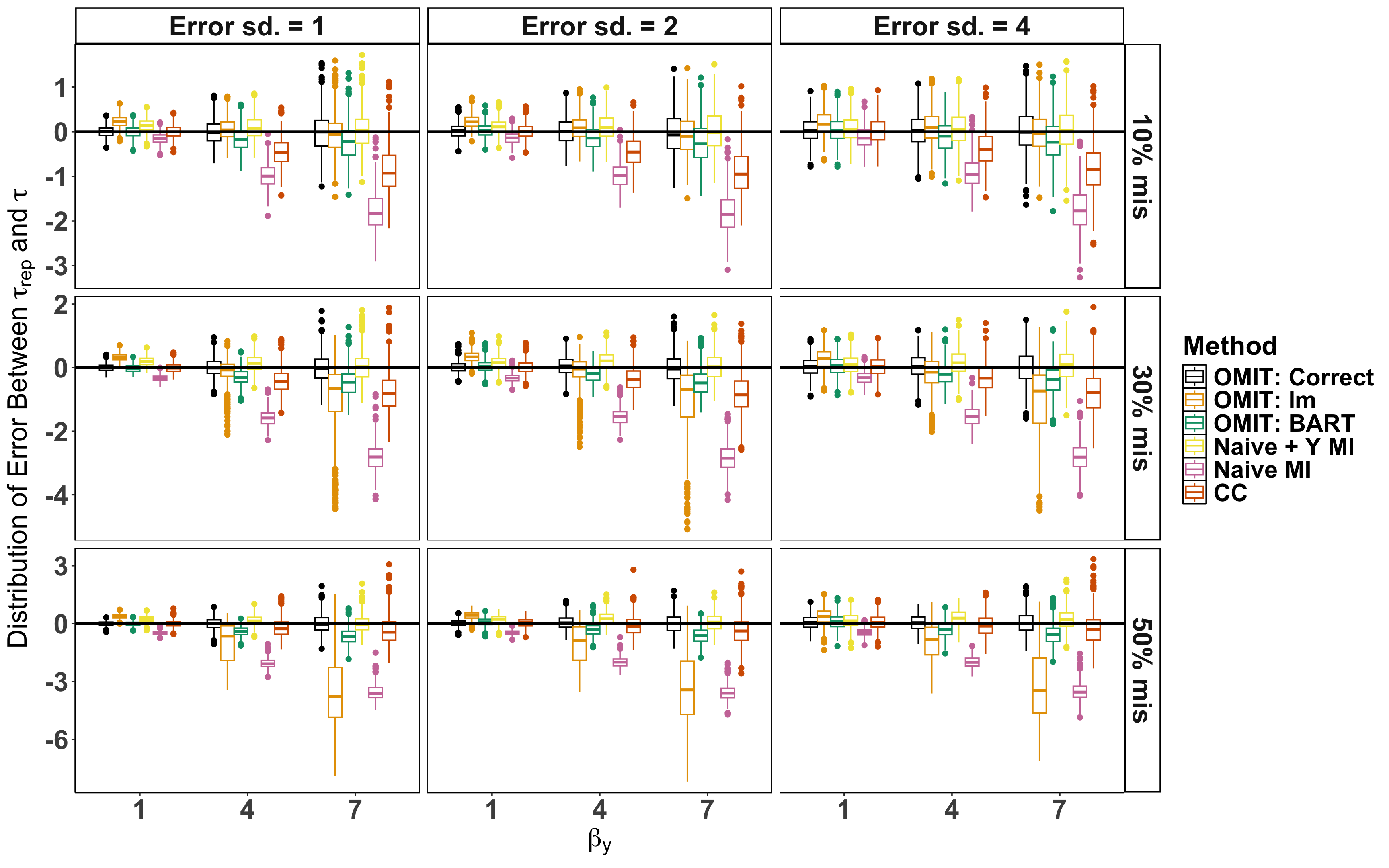}

    \caption{Completing the information in Figure \ref{fig:bias} in the main text to include 50\% missingness as well as the results from using a single model for specification of $p(t_i \mid \boldsymbol x_i, y_i)$ labeled Naive+Y MI.}\label{fig:biasv2}
    \end{figure}

\begin{figure}[H]
    \centering
\includegraphics[width=0.9\linewidth,keepaspectratio]{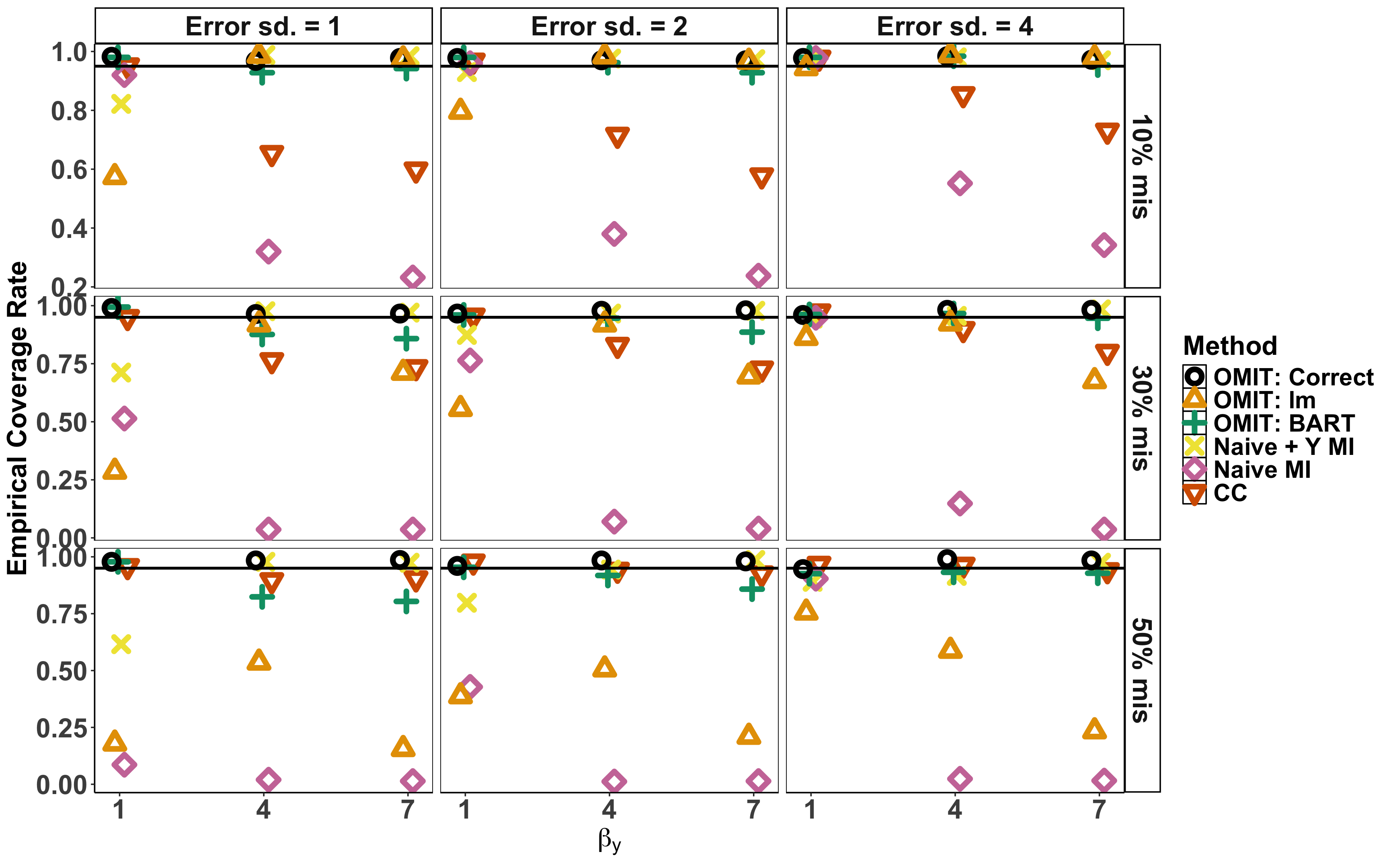}
     \caption{Completing the information in Figure \ref{fig:coverage} in the main text to include 50\% missingness as well as the results from using a single model for specification of $p(t_i \mid \boldsymbol x_i, y_i)$ labeled ``Naive+Y MI.''}\label{fig:coveragev2}
\end{figure}

Naive+Y MI performs competitively to the OMIT approaches in terms of point estimation  and confidence interval coverage rate. This is because there is a monotone relationship between treatment status and $y_i$ under the data generating models; larger $y_i$ correspond, in general, to treated individuals. The default probit model for the single model $p(t_i \mid \boldsymbol x_i, y_i)$ picks up on this.  When this monotonicity is relaxed and we use the same model for $p(t_i \mid \boldsymbol x_i, y_i)$, we see substantially inferior performance, as we show in the next section.

\subsection{Results for Naive+Y MI with Nonlinear Response Surfaces}
As alluded to in Section \ref{sec:conclusion} of the main text, single imputation models for $p(t_i \mid \boldsymbol x_i, y_i)$ may suffer from poor performance when it is challenging to model the relationship between $t$ and $y$ in the presence of $\boldsymbol x$. To illustrate this, we leverage the simulation design in Section \ref{sec:sim} of the main text but change the outcome models \eqref{y1model} and \eqref{y0model} to include cubic functions of $x_{2}$. We have
\begin{align}
     y \mid \boldsymbol x, t = 1 & \sim N(2 + x_1 + \beta_{y}x_{2}^{3}, \sigma^{2})\label{y1modelv2}\\
    y \mid \boldsymbol x, t = 0 & \sim N(1 + x_1 + x_{2}^{3}, \sigma^{2})\label{y0modelv2}
\end{align}
We  use OMIT:Correct adjusted for the new outcome model, OMIT:lm, and OMIT:BART. For Naive+Y MI, we  use a probit regression with main effects for covariates and scaled outcomes. Figure \ref{fig:biasNpY} displays the differences in the point estimates and $\tau$ for all methods.

\begin{figure}[t]
    \centering
\includegraphics[width=0.9\linewidth,keepaspectratio]{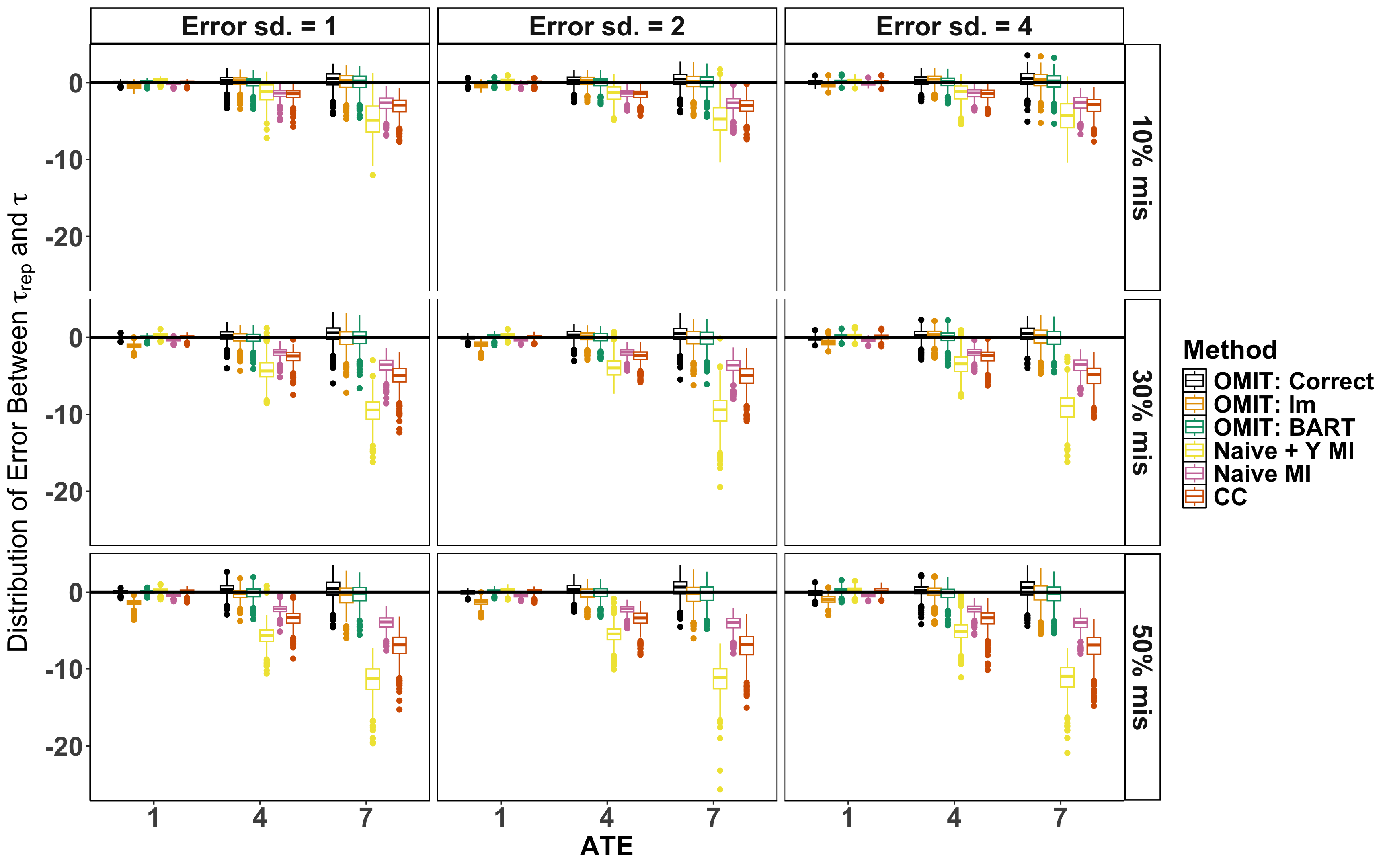}
    \caption{Illustrating the potential pitfalls of using a single model for $p(t_i \mid \boldsymbol x_i, y_i)$.  Naive+Y MI, which leverages the information in the outcomes under one predictive model for the treatments, clearly provides inferior inferences relative to the OMIT approaches.}\label{fig:biasNpY}
    \end{figure}

In this simulation, there is no longer a monotonic relationship between the probability of treatment and the outcomes: both large and small values are likely to be treated. This characteristic is not captured by the Naive+Y outcome model. More generally, it can be difficult to specify one model which adapts to these features. By contrast, all OMIT approaches perform well, demonstrating the modeling flexibility and robustness of the approach. The CC analysis is  biased in this setting due to the dependence of the missingness on the outcomes, which have substantially more variability due to the cubic term in the data generating models.